\documentclass[12pt]{article}

\usepackage{amssymb,amsmath,amsthm,amscd,latexsym}
\usepackage{amsfonts}
\usepackage{tikz}
\usetikzlibrary{arrows,positioning,shapes,fit,calc}
\usepackage{tikz-cd}
\usepackage[mathscr]{eucal}
\usepackage{color,hyperref}
\usepackage{graphicx}
\usepackage{multirow}
\usepackage[ruled,vlined]{algorithm2e}

 \newtheorem{theorem}{Theorem}[section]
 \newtheorem{corollary}[theorem]{Corollary}
 \newtheorem{lemma}[theorem]{Lemma}
 
 \theoremstyle{definition}
 
 \theoremstyle{remark}
 
 \newtheorem{example}[theorem]{Example}
 \numberwithin{equation}{section}
 \font\msbm=msbm10 at 12pt

\newcommand{\FF}{\mbox{\msbm F}}

\def\vv{\mathbf{v}}

\def\vw{\mathbf{w}}

\def\1v{\mathbf{1}}
\def\0v{\mathbf{0}}

\begin{document}

\title{Group Matrix Ring Codes and Constructions of Self-Dual Codes}
\author{S. T. Dougherty \\ University of Scranton \\ Scranton, PA, 18518,  USA  \\Adrian Korban \\
Department of Mathematical and Physical Sciences \\
University of Chester\\
Thornton Science Park, Pool Ln, Chester CH2 4NU, England \\
Serap \c{S}ahinkaya \\
Tarsus University, Faculty of Engineering \\ Department of Natural and Mathematical Sciences \\
Mersin, Turkey \\
Deniz Ustun \\
Tarsus University, Faculty of Engineering \\ Department of Computer Engineering \\
Mersin, Turkey}
 \maketitle

\begin{abstract}
In this work, we study codes generated by elements that come from group matrix rings. We present a matrix construction which we use to generate codes in two different ambient spaces: the matrix ring $M_k(R)$ and the ring $R,$ where $R$ is the commutative Frobenius ring. We show that codes over the ring $M_k(R)$ are one sided ideals in the group matrix ring $M_k(R)G$ and the corresponding codes over the ring $R$ are $G^k$-codes of length $kn.$ Additionally, we give a generator matrix for self-dual codes, which consist of the mentioned above matrix construction. We employ this generator matrix to search for binary self-dual codes with parameters $[72,36,12]$ and find new singly-even   and doubly-even codes of this type. In particular, we construct $16$ new Type~I and $4$ new Type~II binary $[72,36,12]$ self-dual codes.
\end{abstract}

\section{Introduction}\label{intro}

Self-dual codes are one of the most widely studied and interesting class of codes.  They have been shown to have strong connections to unimodular lattices, invariant theory, and designs.  In particular, binary self-dual codes have been extensively studied and numerous construction techniques of self-dual codes have been used in an attempt to find optimal self-dual codes.

In this work, we give a new construction of self-dual codes motivated by the constructions given in \cite{45} and \cite{7}.
In our construction, we use group rings where the ring is a ring of matrices to construct generator matrices of self-dual codes.   The main point of this construction is to find codes that other techniques have missed.   We construct numerous new self-dual codes using this technique.

We begin with some definitions.   A code $C$ over an alphabet $A$ of length $n$ is a subset of $A^n$.  We say that the code is linear over $A$ if $A$ is a ring and $C$ is  a submodule.  This implies that when $A$ is a finite field, then $C$ is a vector space.  We attach to the ambient space the standard Euclidean inner-product, that is $[\vv,\vw] = \sum v_i w_i$. When $A$ is commutative,  we define the orthogonal to this inner-product as $C^\perp = \{ \vw \ |  [\vw,\vv]=0, \forall \vv \in C \}.$  If the ring $A$ is not commutative, then we say that the code is either left linear or right linear depending if it  is a left or right module.    In this scenario, we have two orthogonals, namely ${\cal L}(C) = \{ \vw \ | \ [\vw,\vv]=0, \forall \vv \in C \}$ and ${\cal R}(C) = \{ \vw \ | \ [\vv,\vw]=0, \forall \vv \in C \}.$  If the ring is not commutative then these two codes are not necessarily equal, and in general will not be.  Moreover, ${\cal L}(C)$ is a left linear code and ${\cal R}(C)$ is a right linear code.   If the ring is commutative, then ${\cal L}(C) = {\cal R}(C) = C^\perp.$  It is known that if $C$ is a left linear code over a Frobenius ring $R$ then
$|C| |{\cal R}(C) | = |A^n|$ and if $C$ is a right linear code over a Frobenius ring $R$ then
$|C| |{\cal L}(C) | = |A^n|.$  For commutative rings, this gives that $|C||C^\perp|=|A^n|$ or ${\rm dim}(C) + {\rm dim} (C^\perp) =n$ as usual.  For a complete description of codes over commutative rings see \cite{Doughertybook}. For a description of codes over non-commutative rings see \cite{noncom}.  Throughout this work we assume that every ring has a multiplicative identity and is finite.

An upper bound on the minimum Hamming distance of a binary self-dual code
was given in \cite{AIX}. Specifically, let $d_{I}(n)$ and $d_{II}(n)$ be the
minimum distance of a Type~I (singly-even) and Type~II (doubly-even) binary code of length $n$,
respectively. Then
\begin{equation*}
d_{II}(n) \leq 4\lfloor \frac{n}{24} \rfloor+4
\end{equation*}
and
\begin{equation*}
d_{I}(n)\leq
\begin{cases}
\begin{matrix}
4\lfloor \frac{n}{24} \rfloor+4 \ \ \ if \ n \not\equiv 22 \pmod{24} \\
4\lfloor \frac{n}{24} \rfloor+6 \ \ \ if \ n \equiv 22 \pmod{24}.%
\end{matrix}%
\end{cases}%
\end{equation*}
Self-dual codes meeting these bounds are called \textsl{extremal}.

In this work, we shall use the theory of group rings to build codes.  We shall give the necessary definitions for this study. Let $R$ be a ring, then if $R$ has an identity $1_R,$ we say that $u \in R$ is a unit in $R$ if and only if there exists an element $w \in R$ with $uw=1_R.$
Let $G$ be a finite group of order $n$, then the group ring $RG$ consists of $\sum_{i=1}^n \alpha_i g_i$, $\alpha_i \in R$, $g_i \in G.$  

Addition in the group ring is done by coordinate addition, namely $$\sum_{i=1}^n \alpha_i g_i + \sum_{i=1}^n \beta_i g_i = 
\sum_{i=1}^n (\alpha_i + \beta_i ) g_i.$$ The product of two elements in  a group ring  is given by 
$$(\sum_{i=1}^n \alpha_i g_i)( \sum_{j=1}^n \beta_j g_j)  = \sum_{i,j} \alpha_i \beta_j g_i g_j .$$ This gives that the coefficient of $g_k$ in the product is $ \sum_{g_i g_j = g_k } \alpha_i \beta_j .$  
Notice, that while group rings can use rings and groups of arbitrary cardinality, we restrict ourselves to finite groups and finite rings.  Note that we have not assumed that group nor the ring is commutative.

The space of $k$ by $k$ matrices with coefficients in the ring $R$ is denoted by $M_k(R).$ 
It is immediate that $M_k(R)$ is a ring, however, it is, in general, a non-commutative ring.   Moreover, it is fundamental in the study of non-commutative rings since any finite ring moded out by its Jacobson radical is isomorphic to a direct product of matrix rings.  Moreover, we know that $M_k(R)$ is a Frobenius ring, when $R$ is Frobenius.

A circulant matrix is one where each row is shifted one element to the right relative to the preceding row. We label the circulant matrix as $A=CIRC(\alpha_1,\alpha_2\dots , \alpha_n),$ where $\alpha_i$ are ring elements. A block-circulant matrix is one where each row contains blocks which are square matrices. The rows of the block matrix are defined by shifting one block to the right relative to the preceding row. We label the block-circulant matrix as $\mbox{CIRC}(A_1,A_2,\dots A_n),$ where $A_i$ are $k \times k$ matrices over the ring $R.$ The transpose of a matrix $A,$ denoted by $A^T,$ is a matrix whose rows are the columns of $A,$ i.e., $A^T_{ij}=A_{ji}.$  A symmetric matrix is a square matrix that is equal to its transpose. A persymmetric matrix is a square matrix which is symmetric with respect to the north-east-to-south-west diagonal.

\section{Matrix Construction from Group Matrix Rings}

The following construction of a matrix  which was used to construct codes that were ideals in a group ring was first given for codes over fields
by Hurley in \cite{7}. It was then extended to finite commutative Frobenius rings in \cite{45}. Let $
R$ be a finite commutative Frobenius ring and let $G=\{g_1,g_2,\dots,g_n\}$
be a group of order $n$. Let $v=\alpha_{g_1}g_1+\alpha_{g_2}g_2+\dots
+\alpha_{g_n}g_n \in RG.$ Define the matrix $\sigma(v) \in M_n(R)$ to be
\begin{equation}\label{sigmav}
\sigma(v)=%
\begin{pmatrix}
\alpha_{g_1^{-1}g_1} & \alpha_{g_1^{-1}g_2} & \alpha_{g_1^{-1}g_3} & \dots &
\alpha_{g_1^{-1}g_n} \\
\alpha_{g_2^{-1}g_1} & \alpha_{g_2^{-1}g_2} & \alpha_{g_2^{-1}g_3} & \dots &
\alpha_{g_2^{-1}g_n} \\
\vdots & \vdots & \vdots & \vdots & \vdots \\
\alpha_{g_n^{-1}g_1} & \alpha_{g_n^{-1}g_2} & \alpha_{g_n^{-1}g_3} & \dots &
\alpha_{g_n^{-1}g_n}%
\end{pmatrix}%
.
\end{equation}

We note that the elements $g_1^{-1}, g_2^{-1}, \dots, g_n^{-1}$ are the
elements of the group $G$ in a some given order.

This matrix was used as a generator matrix for codes.  The form of the matrix guaranteed that the resulting code would correspond to an ideal in the group ring and thus have the group $G$ as a subgroup of its automorphism group, that is the group $G$, acting on the coordinates, would leave the code fixed.  The fundamental purpose of this was to construct codes that were not found using more traditional construction techniques.  It was shown in \cite{45} that certain classical constructions would only produce a subset of all possible codes (in this case self-dual codes) and would often miss codes that were of particular interest.  With this in mind we are interested in expanding these kinds of constructions to enable us to find codes that would be missed with other construction techniques.

We now generalize the matrix construction just defined.   Let $R$ be a finite commutative ring and let $G=\{g_1,g_2,\dots,g_n\}$
be a group of order $n$.  We note that  no assumption about the groups commutativity is made.  Let $v=A_{g_1}g_1+A_{g_2}g_2+\dots+A_{g_n}g_n \in M_k(R)G,$ that is, each $A_{g_i}$ is a $k \times k$ matrix with entries from the ring $R.$ Define the block matrix $\sigma_k(v) \in (M_{k}(R))_n$ to be

\begin{equation}\label{sigmakv}
\sigma_k(v)=
\begin{pmatrix}
A_{g_1^{-1}g_1} & A_{g_1^{-1}g_2} & A_{g_1^{-1}g_3} & \dots &
A_{g_1^{-1}g_n} \\
A_{g_2^{-1}g_1} & A_{g_2^{-1}g_2} & A_{g_2^{-1}g_3} & \dots &
A_{g_2^{-1}g_n} \\
\vdots & \vdots & \vdots & \vdots & \vdots \\
A_{g_n^{-1}g_1} & A_{g_n^{-1}g_2} & A_{g_n^{-1}g_3} & \dots &
A_{g_n^{-1}g_n}
\end{pmatrix}
.
\end{equation}

We note that the element $v$ is an element of the group matrix ring $M_k(R)G.$
Of course, this is the same construction as was previously given for group rings but we specify it here since we will use it in very different ways.  Namely, we can consider the matrix as generating two distinct codes in different ambient spaces.

This group matrix ring can be non-commutative in two ways.  First, since the group may not be commutative, multiplication on the left by an element $g \in G$ can give a different element than multiplication on the right by $g$.  We note that in generating the matrix $\sigma_k(v)$ the rows are formed from elements that were constructed by a group element multiplying on the left.   Moreover, multiplication by an element $B \in M_k(R)$ on the left   can give a different element than multiplication on the right by $B$.

As in the matrix $\sigma(v)$ from Equation~\ref{sigmav}, the elements $g_1^{-1}, g_2^{-1}, \dots, g_n^{-1}$ are the
elements of the group $G$ given in a some order.  This order is used in order aid in the computational aspects of some proofs.    We note that when $k=1$ then $\sigma_1(v)=\sigma(v),$ that is, $\sigma_1(v)$ is equivalent to the matrix $\sigma(v)$ in the original definition.  In general, we shall often assume that $k>1.$

The next theorem sets up some useful algebraic tools.  

\begin{theorem}\label{sigmakisringhomomorphism}
Let $R$ be a finite commutative ring. 
Let $G$ be a group of order $n$ with a fixed listing of its elements. Then the map $ \sigma_k: M_k(R)G  \rightarrow M_n(M_k(R))$ is a bijective matrix ring homomorphism.
\begin{proof}
Let $G=\{g_1,g_2,\dots,g_n\}$ be the listing of the elements of $G.$ Now define the map $\sigma_k : M_k(R)G \rightarrow M_n(M_k(R))$ as follows. Suppose $v=\sum_{i=1}^n A_{g_i}g_i.$ Then
$$\sigma_k(v)=
\begin{pmatrix}
A_{g_1^{-1}g_1} & A_{g_1^{-1}g_2} & A_{g_1^{-1}g_3} & \dots &
A_{g_1^{-1}g_n} \\
A_{g_2^{-1}g_1} & A_{g_2^{-1}g_2} & A_{g_2^{-1}g_3} & \dots &
A_{g_2^{-1}g_n} \\
\vdots & \vdots & \vdots & \vdots & \vdots \\
A_{g_n^{-1}g_1} & A_{g_n^{-1}g_2} & A_{g_n^{-1}g_3} & \dots &
A_{g_n^{-1}g_n}
\end{pmatrix}$$
where each $A_{g_i}$ is a square matrix of order $k.$ It can be easily verified that this mapping is additive, surjective and injective. We now show that $\sigma_k$ is multiplicative. Consider $w=\sum_{i=1}^n B_{g_i}g_i$ then
$$\sigma_k(w)=
\begin{pmatrix}
B_{g_1^{-1}g_1} & B_{g_1^{-1}g_2} & B_{g_1^{-1}g_3} & \dots &
B_{g_1^{-1}g_n} \\
B_{g_2^{-1}g_1} & B_{g_2^{-1}g_2} & B_{g_2^{-1}g_3} & \dots &
B_{g_2^{-1}g_n} \\
\vdots & \vdots & \vdots & \vdots & \vdots \\
B_{g_n^{-1}g_1} & B_{g_n^{-1}g_2} & B_{g_n^{-1}g_3} & \dots &
B_{g_n^{-1}g_n}
\end{pmatrix}.$$
Now suppose $w * v=t,$ where $t=\sum_{i=1}^n C_{g_i}g_i.$ Then
$$\sigma_k(w) * \sigma_k(v)=
\begin{pmatrix}
C_{g_1^{-1}g_1} & C_{g_1^{-1}g_2} & C_{g_1^{-1}g_3} & \dots &
C_{g_1^{-1}g_n} \\
C_{g_2^{-1}g_1} & C_{g_2^{-1}g_2} & C_{g_2^{-1}g_3} & \dots &
C_{g_2^{-1}g_n} \\
\vdots & \vdots & \vdots & \vdots & \vdots \\
C_{g_n^{-1}g_1} & C_{g_n^{-1}g_2} & C_{g_n^{-1}g_3} & \dots &
C_{g_n^{-1}g_n}
\end{pmatrix}$$
and this is $\sigma_k(t)=\sigma_k(w * v)$ as required.
\end{proof}
\end{theorem}

Let  the first column of the matrix in Equation~(\ref{sigmakv}) be labelled by $g_1,$ the second column by $g_2,$ etc. Then if $b=\sum_{i=1}^n B_{g_i}g_i$ is in $M_k(R)G$ then the coefficient of $g_i$ in the product $b * v$ is $(B_{g_1},B_{g_2},\dots,B_{g_n})$ times the i-th column of $\sigma_k(v).$  We note here that we are multiplying by $b$ on the left.   This could easily be done on the right to get a similar result.  

We define the $k$ by $k$ matrix  $I_k$ in the usual way.  That is $(I_k)_{ij} =1$ if $i=j$ and $(I_k)_{ij} =0$ if $i \neq j.$
Since each ring we consider in this paper has a multiplicative identity this matrix is always an element in $M_k(R).$

\begin{theorem}
Let $R$ be a finite commutative ring.   Then $v \in M_k(R)G$ is a unit in $M_k(R)G$ if and only if $\sigma_k(v)$ is a unit in $M_n(M_k(R)).$
\begin{proof}
Suppose $v$ is a unit in $M_k(R)G$ and that $w$ is its inverse. Then $v * w=(I_k)_{M_k(R)G}$ and hence $\sigma_k(v * w)=\sigma_k((I_k)_{M_k(R)}G)=I_{kn},$ the identity matrix in $M_n(M_k(R))  .$ Thus $\sigma_k(v) * \sigma_k(w)=I_{kn}.$ Similarly, $\sigma_k(w) * \sigma_k(v)=I_{kn}$ and so $\sigma_k(v)$ is invertible in $M_n(M_k(R))  .$

Suppose now that $\sigma_k(v)$ is a unit in $M_n(M_k(R))  $ and let $N$ denote its inverse. Let $v=\sum_{i=1}^n A_{g_i}g_i.$ Then
$$\sigma_k(v)=
\begin{pmatrix}
A_{g_1^{-1}g_1} & A_{g_1^{-1}g_2} & A_{g_1^{-1}g_3} & \dots &
A_{g_1^{-1}g_n} \\
A_{g_2^{-1}g_1} & A_{g_2^{-1}g_2} & A_{g_2^{-1}g_3} & \dots &
A_{g_2^{-1}g_n} \\
\vdots & \vdots & \vdots & \vdots & \vdots \\
A_{g_n^{-1}g_1} & A_{g_n^{-1}g_2} & A_{g_n^{-1}g_3} & \dots &
A_{g_n^{-1}g_n}
\end{pmatrix}$$
where each $A_{g_i}$ is a square matrix of order $k.$ Let $(B_1,B_2,\dots,B_n)$ be the first row of $N,$ where $B_i$ are the square matrices each of order $k.$ Then:
\begin{equation}
\begin{matrix}
B_1A_{g_{1}^{-1}g_1}& +& B_2 A_{g_{2}^{-1}g_1}& +& \dots &+& B_n A_{g_{n}^{-1}g_1}& =& I_k, \\

B_1 A_{g_{1}^{-1}g_2}& +& B_2 A_{g_{2}^{-1}g_2}& +& \dots &+& B_n A_{g_{n}^{-1}g_2}& =& \mathbf{0}, \\
\vdots & \vdots & \vdots & \vdots & \vdots & \vdots & \vdots & \vdots &\vdots  \\
B_1 A_{g_{1}^{-1}g_n}& +& B_2 A_{g_{2}^{-1}g_n}& +& \dots &+& B_n A_{g_{n}^{-1}g_n}& =& \mathbf{0}.
\end{matrix}
\end{equation}

Now $v=A_{g_1}g_1+A_{g_2}g_2+\dots+A_{g_n}g_n=A_{g_{i}^{-1}g_1}g_{i}^{-1}g_1+A_{g_{i}^{-1}g_2}g_{i}^{-1}g_2+\dots +A_{g_{i}^{-1}g_n}g_{i}^{-1}g_n,$ for each $i,$ $1 \leq i \leq n.$

Define $w=B_1g_{1}+B_2g_{2}+\dots+B_ng_{n}.$ Then:

$$B_ig_i(A_{g_1}g_1+A_{g_2}g_2+\dots +A_{g_n}g_n)=B_ig_i A_{g_{i}^{-1}g_1}g_{i}^{-1}g_1+B_ig_{i} A_{g_{i}^{-1}g_2}g_{i}^{-1}g_2+$$
$$+\dots +B_ig_{i} A_{g_{i}^{-1}g_n}g_{i}^{-1}g_n=B_i A_{g_{i}^{-1}g_1}g_1+B_i A_{g_{i}^{-1}g_2}g_2+\dots +B_i A_{g_{i}^{-1}g_n}g_n.$$
Hence: $v * w=(B_1g_{1}+B_2g_{2}+\dots+B_ng_{n})(A_{g_1}g_1+A_{g_2}g_2+\dots+A_{g_n}g_n)$ equals to:

$$\begin{matrix}
 &B_1 A_{g_{1}^{-1}g_1}g_1& +& B_2 A_{g_{2}^{-1}g_1}g_1& +& \dots &+& B_n A_{g_{n}^{-1}g_1}g_1&  \\

+ &B_1 A_{g_{1}^{-1}g_2}g_2& +& B_2 A_{g_{2}^{-1}g_2}g_2& +& \dots &+& B_n A_{g_{n}^{-1}g_2}g_2&  \\
\vdots &\vdots & \vdots & \vdots & \vdots & \vdots & \vdots & \vdots   \\
+ &B_1 A_{g_{1}^{-1}g_n}g_n& +& B_2 \alpha_{g_{2}^{-1}g_n}g_n& +& \dots &+& B_n A_{g_{n}^{-1}g_n}g_n& 
\end{matrix}$$
and this is $g_1$ from the above. Thus $g_1^{-1} * w$ is the inverse of $v$ and $v$ is a unit in $M_k(R).$
\end{proof}
\end{theorem}

\section{Group Matrix Ring Codes}

In this section, we employ the matrix construction from the previous section to generate codes in two different ambient spaces.
We  make two distinct constructions.  

{\bf Construction 1} 
For a given element $v \in M_k(R)G,$ we define the following code over the matrix ring $M_k(R)$:
\begin{equation}
C_k(v)=\langle \sigma_k(v) \rangle.
\end{equation}
Here the code is generated by taking the all left linear combinations of the rows of the matrix with coefficients in $M_k(R).$ 

{\bf Construction 2} 
For a given element $v \in M_k(R)G,$ we define the following  code over the  ring $R$.
Construct the matrix $\tau_k(v)$ by viewing each element in a $k$ by $k$ matrix as an element in the larger matrix.  
\begin{equation}
B_k(v)=\langle \tau_k(v) \rangle.
\end{equation}
Here the code $B_k(v)$ is formed by taking all linear combinations of the rows of the matrix with coefficients in $R$.  
In this case the ring over which the code is defined is commutative so it is both a left linear and right linear code.  

The following lemma is immediate.  

\begin{lemma}  Let $R$ be a finite Frobenius ring and let $G$ be a group of order $n$.  Let $v \in M_k(R)G$.  
\begin{enumerate}
\item The matrix $\sigma_k(v)$ is an $n$ by $n$ matrix with elements from $M_k(R)$ and the code $C_k(v)$ is a length $n$ code over $M_k(R)$.  
\item The matrix $\tau_k(v)$ is an $nk$ by $nk$ matrix with elements from $R$ and the code $B_k(v)$ is a length $nk$ code over $R$.  
\end{enumerate} 
\end{lemma}

We illustrate these construction techniques in the following example.

\begin{example}
Let $$v=\begin{pmatrix}
0&0 \\
0&0
\end{pmatrix}+\begin{pmatrix}
0&0 \\
0&0
\end{pmatrix}a+\begin{pmatrix}
0&0 \\
0&0
\end{pmatrix}a^2+\begin{pmatrix}
0&1 \\
1&0
\end{pmatrix}a^3+\begin{pmatrix}
0&1 \\
1&0
\end{pmatrix}b+$$
$$+\begin{pmatrix}
1&1 \\
1&1
\end{pmatrix}ba+\begin{pmatrix}
1&1 \\
1&1
\end{pmatrix}ba^2+\begin{pmatrix}
1&1 \\
1&1
\end{pmatrix}ba^3 \in M_2(\mathbb{F}_2)D_8,$$ where the group $\langle a,b \rangle \cong D_8,$ the dihedral group with $8$ elements.  Then 
$\sigma_2(v)$ generates a code $C_2(v)$ which is the ambient space $M_2(\FF_2)^8.$

The matrix 

$$\small  \tau_2(v)=\begin{pmatrix}
\begin{tabular}{cccccccccccccccc}
0&0&0&0&0&0&0&1& 0&1&1&1&1&1&1&1\\
0&0&0&0&0&0&1&0& 1&0&1&1&1&1&1&1\\
0&1&0&0&0&0&0&0& 1&1&0&1&1&1&1&1\\
1&0&0&0&0&0&0&0& 1&1&1&0&1&1&1&1\\
0&0&0&1&0&0&0&0& 1&1&1&1&0&1&1&1\\
0&0&1&0&0&0&0&0& 1&1&1&1&1&0&1&1\\
0&0&0&0&0&1&0&0& 1&1&1&1&1&1&0&1\\
0&0&0&0&1&0&0&0& 1&1&1&1&1&1&1&0\\

0&1&1&1&1&1&1&1& 0&0&0&1&0&0&0&0\\
1&0&1&1&1&1&1&1& 0&0&1&0&0&0&0&0\\
1&1&0&1&1&1&1&1& 0&0&0&0&0&1&0&0\\
1&1&1&0&1&1&1&1& 0&0&0&0&1&0&0&0\\
1&1&1&1&0&1&1&1& 0&0&0&0&0&0&0&1\\
1&1&1&1&1&0&1&1& 0&0&0&0&0&0&1&0\\
1&1&1&1&1&1&0&1& 0&1&0&0&0&0&0&0\\
1&1&1&1&1&1&1&0& 1&0&0&0&0&0&0&0\\

\end{tabular}
\end{pmatrix}$$
and $\tau_2(v)$ can be row reduced to  

$$\small \begin{pmatrix}
\begin{tabular}{cccccccccccccccc}
0&0&0&0&0&0&0&1& 0&1&1&1&1&1&1&1\\
0&0&0&0&0&0&1&0& 1&0&1&1&1&1&1&1\\
0&1&0&0&0&0&0&0& 1&1&0&1&1&1&1&1\\
1&0&0&0&0&0&0&0& 1&1&1&0&1&1&1&1\\
0&0&0&1&0&0&0&0& 1&1&1&1&0&1&1&1\\
0&0&1&0&0&0&0&0& 1&1&1&1&1&0&1&1\\
0&0&0&0&0&1&0&0& 1&1&1&1&1&1&0&1\\
0&0&0&0&1&0&0&0& 1&1&1&1&1&1&1&0
\end{tabular}
\end{pmatrix}.$$
It can be easily checked that $B_2(v)$ is a binary self-dual code with parameters $[16,8,4]$.
\end{example}

It is clear that the aim of this construction is to construct interesting binary self-dual codes. That is, over the matrix ring, the code constructed was trivial, however, over the binary field the code constructed was an interesting self-dual code.

It is apparent that this generalization opens up a new direction for new constructions of codes. This is because the matrix $\tau_k(v)$ does not only depend on the ring elements and the finite group $G$ as does the matrix $\sigma_k(v)$, but rather $\tau_k(v)$ also depends on the form of the matrices $A_{g_i}.$ We note that the $k \times k$ matrices $A_{g_i}$ over $R$ can each take a different form - this is the first advantage of our generalization over the matrix $\sigma(v)$. We also note that the matrix $\sigma_k(v)$ gives us more freedom for controlling the search field when finding a special family of codes since the matrices $A_{g_i}$ do not have to be fully defined by the ring elements appearing in their first rows - and this is the second advantage of our generalization over the matrix $\sigma(v).$

\begin{theorem}
Let $R$ be a finite commutative Frobenius ring, $k$ a positive integer and  $G$ a finite group of order $n.$ Let $v \in M_k(R)G$. Let $I_k(v)$ be the set of elements of $M_k(R)G$ such that $\sum A_ig_i \in I_k(v)$ if and only if $(A_1,A_2,\dots,A_n) \in C_k(v).$ Then $I_k(v)$ is a one-sided ideal in $M_k(R)G$ and in particular, it is a left ideal.
\begin{proof}
Each row of $\sigma_k(v)$ corresponds to an element of the form $hv$ in $M_k(R)G,$ where $h$ is any element of $G.$ That is, the multiplication by $h$ is done from left. The sum of any two elements in $I(v)$ corresponds exactly to the sum of the corresponding elements in $C_k(v)$ and so $I_k(v)$ is closed under addition. 

Now we shall show when the product of an element in $M_k(R)G$ and an element in $I_k(v)$ is in $I_k(v).$  Let $w_1=\sum B_ig_i \in M_k(R)G,$ where $B_i$ are the $k \times k$ matrices. Then if $w_2$ is a row in $C_k(v),$ it is of the form $\sum C_jh_jv.$ Then $w_1w_2=\sum B_ig_i \sum C_jh_jv= \sum B_iC_jg_ih_jv$ which corresponds to an element in $C_k(v)$ gives that the element is in $I_k(v).$ Next, consider $w_2w_1=\sum C_jh_jv \sum B_ig_i=\sum C_jh_jvB_ig_i$ which may not be an element in $C_k(v).$ Thus, $I_k(v)$ is a left ideal of $M_k(R)G$ and since $M_k(R)$ is a non-commutative matrix ring, we have that $I_k(v)$ is a one-sided ideal of $M_k(R)G.$
\end{proof}
\end{theorem}

Given this theorem, we know that any code $C_k(v)$ has $G$ as a subgroup of its automorphism group.  
This is not true, of course,  for the code $B_k(v)$.

The above two results highlight the difference between group codes studied in \cite{45} and the codes we explore in this work. Namely, in group codes it is the coordinates that are held invariant by the action of the group $G$ and in the codes we study in this work, it is the blocks that are held invariant by the action of the group $G.$ For this reason, from now on, we refer to codes $C_k(v)$ as group matrix ring codes.

Now we show that the orthogonal of a group matrix ring code for some group $G$ is also a group matrix ring code.
Let $J_k$ be a one-sided, left ideal in a group matrix ring $M_k(R)G.$ Define $\mathfrak{R}(C)=\{ w \ | \ vw=0, \forall v \in J_k\}.$ It is immediate that $\mathfrak{R}(J_k)$ is a one-sided, right ideal of $M_k(R)G.$

Let $v=A_{g_1}g_1+A_{g_2}g_2+\dots+A_{g_n}g_n \in M_k(R)G$ and $C_k(v)$ be the corresponding group matrix ring code. Let $\Psi : M_k(R)G \rightarrow (M_k(R))^n$ be the canonical map that sends $A_{g_1}g_1+A_{g_2}g_2+\dots+A_{g_n}g_n$ to $(A_{g_1},A_{g_2},\dots,A_{g_n}).$ Let $J_k$ be the one-sided, left ideal $\Psi^{-1}(C).$ Let $\mathbf{w}=(B_1,B_2,\dots,B_n) \in {\mathfrak R}(C).$ Then

\begin{equation}
[(A_{g_j^{-1}g_1}, A_{g_j^{-1}g_2},\dots,A_{g_j^{-1}g_n}),(B_1,B_2,\dots,B_n)]=0, \ \forall j.
\end{equation}
This gives that

\begin{equation}
\sum_{i=1}^n A_{g_j^{-1}g_i}B_i=0, \ \forall j.
\end{equation}

Let $w=\Psi^{-1}(\mathbf{w})=\sum B_{g_i}g_i$ and define $\overline{\mathbf{w}} \in M_k(R)G$ to be $\overline{\mathbf{w}}=C_{g_1}g_1+C_{g_2}g_2+\dots+C_{g_n}g_n$ where

\begin{equation}
C_{g_i}=B_{g_i^{-1}}.
\end{equation}
Then

\begin{equation}
\sum_{i=1}^n A_{g_j^{-1}g_i}B_i=0 \implies \sum_{i=1}^n A_{g_j^{-1}g_i}C_{g_i^{-1}}=0.
\end{equation}
Then $g_j^{-1}g_ig_i^{-1}=g_j^{-1},$ hence this is the coefficient of $g_j^{-1}$ in the product of $\overline{\mathbf{w}}$ and $g_j^{-1}v.$ This gives that $\overline{\mathbf{w}} \in \mathfrak{R}(J_k)$ if and only if $\mathbf{w} \in \mathfrak{R}(C).$

Let $\phi : (M_k(R))^n \rightarrow M_k(R)G$ by $\phi(\mathbf{w})=\overline{\mathbf{w}}.$ It is clear that $\phi$ is a bijection between $\mathfrak{R}(C)$ and $\mathfrak{R}(\Psi^{-1}(C)).$

\begin{theorem}
Let $C=C_k(v)$ be a group matrix ring code in $M_k(R)$ formed from the element $v \in M_k(R)G.$ Then $\Psi^{-1}(\mathfrak{R}(C))$ is a one-sided, left ideal of $M_k(R)G.$ Moreover, if $C_k(v)$ is a left-linear matrix ring $G$-code with the elements of the group acting on the left  then $\mathfrak{R}(C_k(v))$ is a right-linear matrix group $G$-code with the elements of the group acting on the right.
\begin{proof}
Follows from the above discussion.  
\end{proof}
\end{theorem}

We can now investigate the situation for the code  $B_k(v).$ We begin with a definition.
Let $G$ be a finite group of order $n$ and $R$ a finite Frobenius commutative ring.   Let  $D$ be a code in $R^{sn}$ where the coordinates can be partitioned into $n$ sets of size $s$ where each set is assigned an element of $G$.
If the code $D$ is held invariant by the action of multiplying the coordinate set marker by every element of $G$ then the code $D$ is called a quasi-group code of index $s$.

\begin{lemma} \label{firstlemma}  Let $R$ be a finite Frobenius ring and let $G$ be a finite group with $v \in M_k(R)$. Then
$B_k(v)$ is a quasi-$G$-code of length $nk$ and index $k$.
\end{lemma} 
\begin{proof}
Let $v \in M_k(R)G$ and let  $\vw$ be a row of the matrix $\tau_k(v)$.   Letting any element in $G$ act on the $k$ coordinates corresponding to the matrices, gives a new row of $\tau_k(v).$  Therefore, the code $B_k(v)$ is a 
quasi-$G$-code of length $nk$ and index $k$.
\end{proof}

Consider a quasi-$G$-code of index $k$.  Then rearranging the coordinates so that the $i$-th coordinates of each group of $k$ coordinates are placed sequentially, then it is easy to see that any $(g_1,g_2,\dots,g_n) \in G^n$ holds the code invariant.  Namely, any quasi-$G$-code of length $kn$ and index $k$ is a $G^k$-code.  This gives the following.

\begin{theorem}Let $R$ be a finite Frobenius ring and let $G$ be a finite group with $v \in M_k(R)$. Then
$B_k(v)$ is a $G^k$ code of length $kn.$  
\end{theorem}   
\begin{proof}
Follows from Lemma~\ref{firstlemma} and the previous discussion.  
\end{proof}

\section{Generator Matrices of the form $[I_{kn} | \tau_k(v)]$ and Self-Dual Codes}

In this section, we investigate constructions of binary self-dual codes from $\tau_k(v).$ 

\begin{lemma}\label{I|taukv}
Let $G$ be a group of order $n$ and $v=A_1g_1+A_2g_2+\dots +A_ng_n$ be an element of the group matrix ring $M_k(R)G.$ The matrix $[I_{kn}|\tau_k(v)]$ generates a self-dual code over $R$ if and only if $\tau_k(v)\tau_k(v)^T=-I_{kn}$.
\end{lemma}
\begin{proof}
Follows from the standard proof that $(I_m \ |  \ A)$ generates a self-dual code of length $2m$ if and only if $AA^T= - I_m$.   
\end{proof}

Recall that the canonical involution $* : RG \rightarrow RG$ on a group ring $RG$ is given by $v^*=\sum_g a_g g^{-1},$ for $v=\sum_g a_g g \in RG.$ Also, recall that there is a connection between $v^*$ and $v$ when we take their images under the map $\sigma,$ given by
\begin{equation}
\sigma(v^*)=\sigma(v)^T.
\end{equation}
The above connection can be extended to the group matrix ring $M_k(R)G.$ Namely, let $* : M_k(R)G\rightarrow M_k(R)G$ be the canonical involution on the group matrix ring $M_k(R)G$ given by $v^*=\sum_g A_gg^{-1},$ for $v=\sum_g A_gg \in M_k(R)G$ where $A_g$ are the $k \times k$ blocks. Then we have the following connection between $v^*$ and $v$ under the map $\tau_k$:
\begin{equation}
\tau_k(v^*)=\tau_k(v)^T.
\end{equation} 

\begin{lemma}\label{taukisringhomomorphism}
Let $R$ be a finite commutative ring. Let $G$ be a group of order $n$ with a fixed listing of its elements. Then the map $\tau_k : v \rightarrow M(R)_{kn}$ is a bijective ring homomorphism.
\begin{proof}
The proof is similar to the proof in Theorem~\ref{sigmakisringhomomorphism} and simply consists of showing that addition and multiplication are preserved.
\end{proof}
\end{lemma}
 
Now, combining together Lemma~\ref{I|taukv}, Lemma~\ref{taukisringhomomorphism} and the fact that $\tau_k(v)=-I_{kn}$ if and only if $v=-I_k,$ we get the following corollary.

\begin{corollary}
Let $M_k(R)G$ be a group matrix ring, where $M_k(R)$ is a non-commutative Frobenius matrix ring. For $v \in M_k(R)G,$ the matrix $[I_{kn}|\tau_k(v)]$ generates a self-dual code over $R$ if and only if $vv^*=-I_k.$ In particular $v$ has to be a unit.
\end{corollary}

When we restrict our attention to a matrix ring of characteristic 2, we have that $-I_k=I_k,$ which leads to the following further corollary:

\begin{corollary}
Let $M_k(R)G$ be a group matrix ring, where $M_k(R)$ is a non-commutative Frobenius matrix ring of characteristic 2. Then the matrix $[I_{kn}|\tau_k(v)]$ generates a self-dual code over $R$ if and only if $v$ satisfies $vv^*=I_k,$ namely $v$ is a unitary unit in $M_k(R)G.$ 
\end{corollary}

\subsection{New binary self-dual codes of length 72}

In this section, we search for binary self-dual codes with parameters $[72,36,12]$ by considering generator matrices of the form $[I_{kn} | \tau_{kn}(v)],$ where $v \in M_k(\mathbb{F}_2)G$ for different values of $k$ and different groups $G$ to show the strength of our construction and particularly, the strength of the matrix $\tau_{kn}(v).$

The possible weight enumerators for a Type~I $[72,36,12]$ codes are as follows (\cite{Dougherty1}):
$$W_{72,1}=1+2\beta y^{12}+(8640-64\gamma)y^{14}+(124281-24\beta+384\gamma)y^{16}+\dots$$
$$W_{72,2}=1+2 \beta y^{12}+(7616-64 \gamma)y^{14}+(134521-24 \beta+384 \gamma)y^{16}+\dots$$
where $\beta$ and $\gamma$ are parameters.
The possible weight enumerators for Type~II $[72,36,12]$ codes are (\cite{Dougherty1}):
$$1+(4398+\alpha)y^{12}+(197073-12\alpha)y^{16}+(18396972+66\alpha)y^{20}+\dots $$
where $\alpha$ is a parameter.

Many codes for different values of $\alpha$, $\beta$ and $\gamma$ have been constructed in \cite{Bouyukliev1, Dontcheva1, Dougherty1, Dougherty2, Gulliver1, Yankov2, Kaya1, Korban1, Yildiz1, Yankov1, Zhdanov1, Zhdanov2}. For an up-to-date list of all known Type~I and Type~II binary self-dual codes with parameters $[72,36,12]$ please see \cite{selfdual72}.

We now split the remaining of this section into subsections, where in each we consider a generator matrix of the form $[I_{kn} | \tau_k(v)]$ for a specific group $G$ and some specific $k \times k$ block matrices to search for binary self-dual codes with parameters $[72,36,12].$ All the upcoming computational results were obtained by performing searches using a particular algorithm technique (see \cite{Korban1} for details) in the software package MAGMA (\cite{MAGMA}).

\subsubsection{The group $C_2$ and $18 \times 18$ block matrices}

In this section, we consider the cyclic group $C_2$ with some $18 \times 18$ block matrices. Let $G=\langle x \ | \ x^2=1 \rangle \cong C_{2}.$ Let $v=\sum_{i=0}^1 Y_ix^i \in (M(\mathbb{F}_2))_{18}C_{2},$ then
\begin{equation}
\tau_{18}(v)=\begin{pmatrix}
Y_0&Y_1\\
Y_1&Y_0
\end{pmatrix},
\end{equation}
where
$$Y_0=\begin{pmatrix}
A&B\\
B&A
\end{pmatrix}, \ Y_1=\begin{pmatrix}
C&D\\
D&C
\end{pmatrix}$$
with $$A=CIRC(A_1,A_2,A_3),$$
$$B=CIRC(A_4,A_5,A_6),$$
$$C=CIRC(A_7,A_8,A_9),$$
$$D=CIRC(A_{10},A_{11},A_{12})$$
where $A_i$ are some matrices. We now employ a generator matrix of the form $[I \ | \ \tau_{18}(v)],$ where $I$ is the $36 \times 36$ identity matrix, for different forms of the matrices $A_i$ to search for binary self-dual codes with parameters $[72,36,12].$ We only list codes with parameters in their weight distributions that were not known in the literature before. Also, since the matrix $\sigma_{18}(v)$ is fully defined by the first row, we only list the first row of the matrices $Y_0,$ and $Y_1$ which we label as $r_{Y_0}$ and $r_{Y_1}$ respectively.

\begin{enumerate}
\item[Case 1.] Here we let
$$A_1=revcirc(a_1,a_2,a_3),$$ 
$$A_2=revcirc(a_4,a_5,a_6),$$
$$\dots,$$
$$A_{12}=revcirc(a_{34},a_{35},a_{36}).$$

\begin{table}[h!]
\caption{New Type~I $[72,36,12]$ Codes}
\resizebox{0.65\textwidth}{!}{\begin{minipage}{\textwidth}
\centering
\begin{tabular}{ccccccc}
\hline
      & Type       & $r_{Y_0}$                                   & $r_{Y_1}$                                   & $\gamma$ & $\beta$ & $|Aut(C_i)|$ \\ \hline
$C_1$ & $W_{72,1}$ & $(0,1,1,0,1,1,0,0,0,0,0,0,0,1,1,0,1,0)$ & $(0,1,0,0,1,0,1,1,0,1,0,0,1,1,1,1,0,1)$ & $36$      & $543$   & $72$         \\ \hline
\end{tabular}
\end{minipage}}
\end{table}

\begin{table}[h!]
\caption{New Type~II $[72,36,12]$ Codes}
\resizebox{0.71\textwidth}{!}{\begin{minipage}{\textwidth}
\centering
\begin{tabular}{ccccc}
\hline
      & $r_{Y_0}$ & $r_{Y_1}$ & $\alpha$ & $|Aut(C_i)|$ \\ \hline
$C_2$ & $(0,1,0,1,0,0,1,0,0,0,1,1,0,1,1,1,0,1)$  & $(1,1,1,1,1,1,1,1,1,0,1,1,1,0,0,1,0,1)$  & $-2604$     & $72$         \\ \hline
\end{tabular}
\end{minipage}}
\end{table}

\item[Case 2.] Here we let 
$$A_1=revcirc(a_1,a_2,a_3),$$
$$A_2=revcirc(a_4,a_5,a_6),$$
$$\dots,$$
$$A_6=revcirc(a_{16},a_{17},a_{18}),$$
$$A_7=circ(a_{19},a_{20},a_{21}),$$
$$A_8=circ(a_{22},a_{23},a_{24}),$$
$$\dots,$$
$$A_{12}=circ(a_{34},a_{35},a_{36}).$$

\begin{table}[h!]
\caption{New Type~I $[72,36,12]$ Codes}
\resizebox{0.65\textwidth}{!}{\begin{minipage}{\textwidth}
\centering
\begin{tabular}{ccccccc}
\hline
      & Type       & $r_{Y_0}$                                   & $r_{Y_1}$                                   & $\gamma$ & $\beta$ & $|Aut(C_i)|$ \\ \hline
$C_3$ & $W_{72,1}$ & $(1,0,0,1,1,0,0,0,1,1,1,0,0,0,0,1,0,1)$ & $(1,1,0,0,1,1,1,0,1,1,0,0,0,1,0,0,1,0)$ & $0$      & $342$   & $36$         \\ \hline
$C_4$ & $W_{72,1}$ & $(1,0,1,1,1,0,1,1,0,0,0,1,1,0,0,1,0,1)$ & $(1,1,0,1,0,0,1,1,1,1,1,0,1,0,1,0,0,1)$ & $18$      & $420$   & $36$         \\ \hline
$C_5$ & $W_{72,1}$ & $(0,0,0,0,0,0,0,1,1,1,1,1,1,0,1,1,1,0)$ & $(1,0,1,1,0,0,1,1,1,0,0,0,0,0,1,0,0,1)$ & $36$      & $561$   & $72$         \\ \hline
\end{tabular}
\end{minipage}}
\end{table}

\begin{table}[h!]
\caption{New Type~II $[72,36,12]$ Codes}
\resizebox{0.71\textwidth}{!}{\begin{minipage}{\textwidth}
\centering
\begin{tabular}{ccccc}
\hline
      & $r_{Y_0}$ & $r_{Y_1}$ & $\alpha$ & $|Aut(C_i)|$ \\ \hline
$C_6$ & $(1,1,0,0,1,0,1,0,0,1,0,0,1,0,0,1,1,1)$  & $(1,1,1,0,0,0,1,1,0,1,0,0,1,0,1,1,1,0)$  & $-2706$     & $36$         \\ \hline
\end{tabular}
\end{minipage}}
\end{table}

\item[Case 3.] Here we let
$$A_1=circ(a_1,a_2,a_3),$$
$$A_2=circ(a_4,a_5,a_6),$$
$$\dots,$$
$$A_6=circ(a_{16},a_{17},a_{18}),$$
$$A_7=revcirc(a_{19},a_{20},a_{21}),$$
$$A_8=revcirc(a_{22},a_{23},a_{24}),$$
$$\dots,$$
$$A_{12}=revcirc(a_{34},a_{35},a_{36}).$$

\begin{table}[h!]
\caption{New Type~I $[72,36,12]$ Codes}
\resizebox{0.65\textwidth}{!}{\begin{minipage}{\textwidth}
\centering
\begin{tabular}{ccccccc}
\hline
      & Type       & $r_{Y_0}$                                   & $r_{Y_1}$                                   & $\gamma$ & $\beta$ & $|Aut(C_i)|$ \\ \hline
$C_7$ & $W_{72,1}$ & $(0,0,0,0,0,0,1,1,0,1,0,0,0,0,1,1,0,1)$ & $(0,0,0,0,0,1,0,1,0,1,0,1,0,1,1,0,1,0)$ & $0$      & $186$   & $36$         \\ \hline
$C_8$ & $W_{72,1}$ & $(0,1,0,1,0,0,1,0,1,1,0,0,1,0,0,0,1,1)$ & $(1,0,1,0,1,1,1,0,1,0,0,0,1,1,0,0,0,1)$ & $18$      & $432$   & $36$         \\ \hline
$C_{9}$ & $W_{72,1}$ & $(1,1,1,0,1,0,1,1,1,0,0,0,1,0,0,0,0,1)$ & $(0,0,1,1,1,1,1,0,1,0,0,0,1,0,0,1,0,0)$ & $36$      & $597$   & $72$         \\ \hline
\end{tabular}
\end{minipage}}
\end{table}

\begin{table}[h!]
\caption{New Type~II $[72,36,12]$ Codes}
\resizebox{0.71\textwidth}{!}{\begin{minipage}{\textwidth}
\centering
\begin{tabular}{ccccc}
\hline
      & $r_{Y_0}$ & $r_{Y_1}$ & $\alpha$ & $|Aut(C_i)|$ \\ \hline
$C_{10}$ & $(1,0,1,1,1,0,0,1,0,0,0,0,0,0,0,0,0,0)$  & $(1,0,1,0,1,1,0,1,0,0,1,1,0,1,0,1,0,1)$  & $-2538$     & $72$         \\ \hline
$C_{11}$ & $(1,0,0,1,0,0,1,1,1,1,1,0,0,0,0,0,1,1)$  & $(0,1,0,1,1,1,1,0,0,1,1,0,1,0,0,0,1,1)$  & $-3066$     & $36$         \\ \hline
\end{tabular}
\end{minipage}}
\end{table}
\end{enumerate}

\subsubsection{The group $D_{18}$ and $2 \times 2$ block matrices}

In this section, we consider the dihedral group $D_{18}$ with some $2 \times 2$ block matrices.

Let $G=\langle x,y \ | \ x^9=y^2=1, x^y=x^{-1} \rangle \cong D_{18}.$ Let $v=\sum_{i=0}^8 \sum_{j=0}^1 A_{1+i+9j}x^iy^j \in M(\mathbb{F}_2)D_{18},$ then
\begin{equation}
\tau_2(v)=\begin{pmatrix}
A&B\\
B^T&A^T
\end{pmatrix},
\end{equation}
with
$$A=CIRC(A_1,A_2,A_3,\dots,A_9),$$
$$B=CIRC(A_{10},A_{11},A_{12},\dots,A_{18})$$
where $A_i$ are some matrices.

We now employ a generator matrix of the form $[I  |  \tau_2(v)],$ where $I$ is the $36 \times 36$ identity matrix, for different forms of the matrices $A_i$ to search for binary self-dual codes with parameters $[72,36,12].$ We only list codes with parameters in their weight distributions that were not known in the literature before. 

\begin{enumerate}
\item[Case 1.] Here we let
$$A_1=circ(a_1,a_2),$$
$$A_2=circ(a_3,a_4),$$
$$\dots,$$
$$A_{12}=circ(a_{35},a_{36}).$$
Since $\tau_2(v)$ is fully defined by the first row, we only list the first rows of the matrices $A$ and $B$ which we label as $r_A$ and $r_B$ respectively.

\begin{table}[htbp]
\caption{New Type~I $[72,36,12]$ Codes}
\resizebox{0.65\textwidth}{!}{\begin{minipage}{\textwidth}
\centering
\begin{tabular}{ccccccc}
\hline
      & Type       & $r_A$                                   & $r_B$                                   & $\gamma$ & $\beta$ & $|Aut(C_i)|$ \\ \hline
$C_{12}$ & $W_{72,1}$ & $(0,1,0,1,0,0,0,1,1,0,1,0,0,0,0,0,0,1)$ & $(0,1,1,1,1,1,1,0,0,0,1,1,1,1,1,0,0,0)$ & $18$      & $237$   & $36$         \\ \hline
$C_{13}$ & $W_{72,1}$ & $(1,1,0,0,0,0,0,0,1,0,1,0,0,0,0,1,1,1)$ & $(1,0,0,1,1,1,0,1,0,0,1,0,1,1,1,0,1,0)$ & $18$      & $387$   & $36$         \\ \hline
$C_{14}$ & $W_{72,1}$ & $(1,0,0,1,0,1,1,1,0,1,1,1,1,1,0,1,0,0)$ & $(0,1,0,1,0,0,0,0,0,1,0,0,1,1,0,0,1,0)$ & $36$     & $417$   & $36$         \\ \hline
$C_{15}$ & $W_{72,1}$ & $(1,1,1,0,1,1,1,1,0,0,0,1,0,0,1,0,1,0)$ & $(0,1,1,0,1,0,1,1,0,1,0,1,0,0,0,0,0,0)$ & $36$     & $564$   & $36$         \\ \hline
\end{tabular}
\end{minipage}}
\end{table}

\item[Case 2.] Here we let
$$A_1=\begin{pmatrix}
a_1&a_2\\
a_3&a_1
\end{pmatrix}, A_2=\begin{pmatrix}
a_4&a_5\\
a_6&a_4
\end{pmatrix}, A_3=\begin{pmatrix}
a_7&a_8\\
a_9&a_7
\end{pmatrix}, \dots , A_9=\begin{pmatrix}
a_{25},a_{26}\\
a_{27},a_{25}
\end{pmatrix}$$
and
$$A_{10}=circ(a_{28},a_{29}),$$
$$A_{11}=circ(a_{30},a_{31}),$$
$$A_{12}=circ(a_{32},a_{33}),$$
$$\dots,$$
$$A_{18}=circ(a_{44},a_{45}).$$
We note here, that the first nine blocks are the $2 \times 2$ per-symmetric matrices - defined by three independent variables and that the next nine blocks are the $2 \times 2$ circulant matrices - defined by two independent variables both appearing in the first rows. This gives a search field of $2^{45}$. To save space, we only list the three variables of each persymmetric matrix which we label as $r_{A_1}, r_{A_2}, r_{A_3}, \dots, r_{A_9}$ and the first row of the matrix $B$ which we label as $r_B$ since this matrix is fully defined by the first row.

\begin{table}[h!]
\caption{New Type~I $[72,36,12]$ Codes}
\resizebox{0.48\textwidth}{!}{\begin{minipage}{\textwidth}
\centering
\begin{tabular}{ccccccccccccccc}
\hline
      & Type       & $r_{A_1}$     & $r_{A_2}$     & $r_{A_3}$     & $r_{A_4}$     & $r_{A_5}$     & $r_{A_6}$     & $r_{A_7}$     & $r_{A_8}$     & $r_{A_9}$     & $r_B$                                   & $\gamma$ & $\beta$ & $|Aut(C_i)|$ \\ \hline
$C_{16}$ & $W_{72,1}$ & $(0,0,0)$ & $(1,1,0)$ & $(0,1,0)$ & $(0,1,1)$ & $(1,0,1)$ & $(0,0,1)$ & $(0,1,1)$ & $(1,0,1)$ & $(1,1,0)$ & $(0,0,0,0,1,1,0,0,0,1,0,0,1,0,1,1,1,1)$ & $9$      & $264$   & $18$         \\ \hline
$C_{17}$ & $W_{72,1}$ & $(0,1,0)$ & $(0,1,0)$ & $(0,0,0)$ & $(0,0,1)$ & $(0,1,0)$ & $(1,0,0)$ & $(0,1,0)$ & $(0,1,0)$ & $(0,1,1)$ & $(1,1,1,1,1,1,0,1,1,0,1,1,1,0,1,1,1,0)$ & $27$     & $345$   & $18$         \\ \hline
\end{tabular}
\end{minipage}}
\end{table}
We note that the code $C_{17}$ in the above table is the first example of a self-dual $[72,36,12]$ code with $\gamma=27$ in its weight distribution.
\end{enumerate}

\subsubsection{The group $C_{6,3}$ and $2 \times 2$ block matrices}

In this section, we consider the cyclic group $C_{6,3}$ and some $2 \times 2$ matrices.

Let $G=\langle x \ | \ x^{6 \cdot 3}=1 \rangle \cong C_{6,3}.$ Let $v=\sum_{i=0}^5 \sum_{j=0}^2 A_{1+i+6j}x^{3i+j} \in M(\mathbb{F}_2)C_{6,3},$ then
\begin{equation}
\tau_2(v)=\begin{pmatrix}
A&B&C\\
C'&A&B\\
B'&C'&A
\end{pmatrix},
\end{equation}
where 
$$A=CIRC(A_1,A_2,\dots,A_6),$$
$$B=CIRC(A_{7},A_{8},\dots,A_{12}),$$
$$C=CIRC(A_{13},A_{14},\dots,A_{18}),$$
$$B'=CIRC(A_{12},A_7,A_8,\dots,A_{11}),$$
$$C'=CIRC(A_{18},A_{13},A_{14},\dots,A_{17})$$
and where $A_i$ are some $2 \times 2$ matrices.

We now employ a generator matrix of the form $[I  |  \tau_2(v)],$ where $I$ is the $36 \times 36$ identity matrix, for different forms of the matrices $A_i$ to search for binary self-dual codes with parameters $[72,36,12].$ We only list codes with parameters in their weight distributions that were not known in the literature before.

\begin{enumerate}
\item[Case 1.] Here we let 
$$A_1=circ(a_1,a_2),$$
$$A_2=circ(a_3,a_4),$$
$$A_3=circ(a_5,a_6),$$
$$\dots,$$
$$A_{18}=circ(a_{35},a_{36}).$$
We note that the search field here is $2^{36}.$ Since the matrix $\tau_2(v)$ is fully defined by the first row, we only list the first rows of the matrices $A, B$ and $C$ which we label as $r_A, r_B$ and $r_C$ respectively.

\begin{table}[h!]
\caption{New Type~I $[72,36,12]$ Codes}
\resizebox{0.65\textwidth}{!}{\begin{minipage}{\textwidth}
\centering
\begin{tabular}{cccccccc}
\hline
      & Type       & $r_A$                       & $r_B$                       & $r_C$                       & $\gamma$ & $\beta$ & $|Aut(C_i)|$ \\ \hline
$C_{18}$ & $W_{72,1}$ & $(1,0,0,0,1,0,0,1,1,1,1,0)$ & $(1,1,0,0,0,1,0,1,1,0,0,0)$ & $(1,1,0,1,0,0,1,1,0,0,0,1)$ & $36$     & $423$   & $72$         \\ \hline
\end{tabular}
\end{minipage}}
\end{table}

\item[Case 2.] Here we let
$$A_1=\begin{pmatrix}
a_1&a_2\\
a_3&a_1
\end{pmatrix}, A_2=\begin{pmatrix}
a_4&a_5\\
a_6&a_4
\end{pmatrix}, A_3=\begin{pmatrix}
a_7&a_8\\
a_9&a_7
\end{pmatrix}, \dots , A_9=\begin{pmatrix}
a_{25},a_{26}\\
a_{27},a_{25}
\end{pmatrix}$$
and 
$$A_{10}=circ(a_{28},a_{29}),$$
$$A_{11}=circ(a_{30},a_{31}),$$
$$A_{12}=circ(a_{32},a_{33}),$$
$$\dots,$$
$$A_{18}=circ(a_{44},a_{45}).$$
We note here, that the first nine blocks are the $2 \times 2$ per-symmetric matrices - defined by three independent variables and that the next nine blocks are the $2 \times 2$ circulant matrices - defined by two independent variables both appearing in the first rows. This gives a search field of $2^{45}$. To save space, we only list the three variables of each persymmetric matrix which we label as $r_{A_1}, r_{A_2}, r_{A_3}, \dots, r_{A_9}$ and the first rows of the matrices $A_{10}, A_{11}, A_{12}, \dots, A_{18}$ which we label as $r_{A_{10}}, r_{A_{11}}, r_{A_{12}}, \dots, r_{A_{18}}$ since these matrices are each defined by the first row.

\begin{table}[htbp]
\caption{New Type~I $[72,36,12]$ Codes}
\resizebox{0.65\textwidth}{!}{\begin{minipage}{\textwidth}
\centering
\begin{tabular}{cccccccccccccc}
\hline
                          & Type                        & $r_{A_1}$    & $r_{A_2}$    & $r_{A_3}$    & $r_{A_4}$    & $r_{A_5}$    & $r_{A_6}$    & $r_{A_7}$    & $r_{A_8}$    & $r_{A_9}$    & $\gamma$              & $\beta$                & $|Aut(C_i)|$          \\ \hline
\multirow{3}{*}{$C_{19}$} & \multirow{3}{*}{$W_{72,1}$} & $(1,1,1)$    & $(0,1,0)$    & $(0,1,1)$    & $(1,0,0)$    & $(1,1,0)$    & $(0,1,1)$    & $(0,0,1)$    & $(1,0,1)$    & $(1,1,1)$    & \multirow{3}{*}{$18$} & \multirow{3}{*}{$342$} & \multirow{3}{*}{$36$} \\
                          &                             & $r_{A_{10}}$ & $r_{A_{11}}$ & $r_{A_{12}}$ & $r_{A_{13}}$ & $r_{A_{14}}$ & $r_{A_{15}}$ & $r_{A_{16}}$ & $r_{A_{17}}$ & $r_{A_{18}}$ &                       &                        &                       \\
                          &                             & $(1,0)$      & $(0,1)$      & $(0,1)$      & $(1,1)$      & $(0,0)$      & $(0,1)$      & $(0,0)$      & $(0,0)$      & $(0,0)$      &                       &                        &                       \\ \hline
                        
\end{tabular}
\end{minipage}}
\end{table}

\begin{table}[htbp]
\resizebox{0.65\textwidth}{!}{\begin{minipage}{\textwidth}
\centering
\begin{tabular}{cccccccccccccc}
\hline
                          & Type                        & $r_{A_1}$    & $r_{A_2}$    & $r_{A_3}$    & $r_{A_4}$    & $r_{A_5}$    & $r_{A_6}$    & $r_{A_7}$    & $r_{A_8}$    & $r_{A_9}$    & $\gamma$              & $\beta$                & $|Aut(C_i)|$          \\ \hline
\multirow{3}{*}{$C_{20}$} & \multirow{3}{*}{$W_{72,1}$} & $(0,0,0)$    & $(0,0,0)$    & $(0,0,1)$    & $(1,0,0)$    & $(0,0,0)$    & $(0,1,0)$    & $(0,0,1)$    & $(0,1,0)$    & $(0,0,0)$    & \multirow{3}{*}{$36$} & \multirow{3}{*}{$510$} & \multirow{3}{*}{$36$} \\
                          &                             & $r_{A_{10}}$ & $r_{A_{11}}$ & $r_{A_{12}}$ & $r_{A_{13}}$ & $r_{A_{14}}$ & $r_{A_{15}}$ & $r_{A_{16}}$ & $r_{A_{17}}$ & $r_{A_{18}}$ &                       &                        &                       \\
                          &                             & $(1,1)$      & $(1,1)$      & $(1,0)$      & $(0,1)$      & $(0,0)$      & $(1,0)$      & $(0,0)$      & $(1,1)$      & $(1,0)$      &                       &                        &                       \\ \hline
                        
\end{tabular}
\end{minipage}}
\end{table}

\end{enumerate}

We would like to stress  that the above constructions represent a very small fraction of the possible matrix constructions that can be derived for the generator matrix $[I_{kn}|\tau_k(v)].$ That is, there are many more different choices for the groups and their sizes, the forms of the $k \times k$ matrices and their sizes which can all lead to constructing optimal binary self-dual codes of various lengths - this shows the strength of our generator matrix.

\section{Conclusion}

In this work, we defined group matrix ring codes that are left ideals in the group matrix ring $M_k(R)G.$ We generalized a well known matrix construction so that this generalization can be used to generate codes in two different ambient spaces. We presented a generator matrix for self-dual codes which we believe can be used to construct many new codes that could not be obtained from other, known in the literature, generator matrices. Additionally, we employed our generator matrix to search for binary self-dual codes. In particular, we constructed Type~I binary $[72,36,12]$ self-dual codes with new weight enumerators in $W_{72,1}$:

\begin{equation*}
\begin{array}{l}
(\gamma =0,\ \  \beta =\{186, 342\}), \\
(\gamma =9,\ \  \beta =\{264\}), \\
(\gamma =18,\ \beta =\{237, 342, 387, 420, 432\}), \\
(\gamma =27,\ \beta =\{345\}), \\
(\gamma =36,\ \beta =\{417, 423, 510, 543, 561, 564, 597\}) \\
\end{array}%
\end{equation*}
and Type~II binary $[72,36,12]$ self-dual codes with new weight enumerators:
\begin{equation*}
\begin{array}{l}
(\alpha =\{-2604, -2706, -2538, -3066\}). \\

\end{array}%
\end{equation*}

A suggestion for future work is to consider the generator matrix we presented in this work, for different groups, different types of the $k \times k$ matrices and different alphabets to search for new optimal binary self-dual codes of different lengths.

\end{document}